\documentclass[article,twocolumn,10pt]{IEEEtran}
\usepackage{graphicx}
\usepackage{float}
\usepackage{amsfonts}
\usepackage{amsmath}
\usepackage{amssymb}
\usepackage{cite}
\usepackage{bm}
\usepackage{mathrsfs}
\usepackage{epsfig}
\usepackage{algorithm}
\usepackage{algorithmic}
\usepackage{amsthm}
\usepackage{color}
\usepackage{graphicx}
\usepackage{caption}
\usepackage{subcaption}

\newtheorem{theorem}{Theorem}

\newtheorem{proposition}{Proposition}
\newtheorem{corollary}{Corollary}
\newtheorem{definition}{Definition}

\begin{document}

\title{Wireless Compressive Sensing Over Fading Channels with Distributed Sparse Random Projections }

\author{\authorblockA{Thakshila Wimalajeewa  \emph{Member, IEEE}, and Pramod K.
Varshney, \emph{Fellow IEEE}}
}

\maketitle\thispagestyle{empty}

%
%

%

\begin{abstract}
We address the problem of recovering a sparse signal observed by a  resource constrained wireless sensor network under channel fading.
Sparse random matrices are  exploited  to reduce the communication cost in forwarding information to a fusion center.  The presence of channel fading leads to inhomogeneity and non Gaussian statistics  in the effective measurement matrix that relates the  measurements  collected  at the fusion center and the sparse signal   being observed.
We analyze the impact of  channel fading on \emph{nonuniform} recovery  of a given  sparse signal by leveraging the properties of heavy-tailed  random matrices.
We quantify the additional number of measurements required to ensure reliable signal recovery in the presence of nonidentical fading channels compared  to that is required with identical  Gaussian channels. Our analysis  provides insights into how to control  the probability of sensor transmissions at each node based on the channel fading statistics  in order to minimize the number of measurements collected at the fusion center for   reliable sparse signal recovery. We further discuss recovery guarantees of a given sparse signal  with any random projection matrix where the elements are sub-exponential with a given sub-exponential norm.   Numerical results are provided to corroborate the theoretical findings.
\end{abstract}
EDICS: ADEL-DIP, CNS-SPDCN

\footnotetext[1]{The authors   are  with the Department of Electrical Engineering and Computer Science, Syracuse University, Syracuse NY 13244. Email: twwewelw@syr.edu, varshney@syr.edu.  This work is  supported by the National Science
Foundation (NSF) under Grant No. 1307775.}

\section{Introduction}
Consider  a wireless sensor network (WSN)   deployed to observe   a compressible signal. The goal is to reconstruct the observed  signal at a distant  fusion center utilizing available network resources efficiently. In order to reduce the energy consumption while forwarding observations to a fusion center, some  preprocessing  is desired so that the fusion center has access to only informative data just querying only a subset of sensors.  Use of compressive sensing (CS) techniques for   compressible data processing  in wireless sensor networks has attracted  attention in the recent literature   \cite{haupt4,wang_ISPN1,Shen_Sensor1,Yang_TSP1,Baron1,Ling_TSP1,
Caione_II1,Sartipi_DCC1,thakshila_icassp1,thakshila_TSP_14,yujiao_icassp1,Fazel_TWC13,Colonnese_EURA13}.
 In \cite{haupt4}, the authors have proposed a multiple access channel (MAC) communication architecture so that the fusion center receives a  compressed version (represented by a  low dimensional linear transformation)  of the original signal observed at multiple nodes. According to that model, the corresponding linear operator is a dense random matrix. Thus,   almost all the sensors in  the network have to participate in forwarding observations consuming a large amount of energy.      The application of sparse random matrices to reduce the communication burden for wireless compressive sensing (WCS) has been addressed by several authors  \cite{wang_ISPN1,Shen_Sensor1,Yang_TSP1} so that not all the sensors forward  observations.  In \cite{yujiao_icassp1}, the authors provide a probabilistic sensor management scheme for target tracking in a WSN exploiting sparse random matrices. In these approaches, the  sparse random matrix is considered to be a { sparse Rademacher}   matrix in which elements may take values  $(+1,0,-1)$ with desired probabilities.   The use of sparse random matrices instead of dense matrices in signal recovery in a general framework (not necessarily in sensor networks)  has been further discussed in several works \cite{Achlioptas_J1,wang5,Gilbert_Proc1,P_Li_KDD1}.

 In practical communication networks, the communication channels between sensor nodes and  the fusion center undergo fading. The presence of fading  affects the recovery capabilities since it leads to inhomogeneity and non Gaussian statistics  in measurement matrices.  In  \cite{Yang_TSP1}, the problem of sparse signal recovery in the presence of fading is addressed where the authors  provide \emph{uniform} recovery guarantees based on restricted isometry property (RIP) considering   sparse Bernoulli matrices.   Two kinds of recovery guarantees with low dimensional random projection matrices  are widely discussed in the CS  literature \cite{candes5,donoho1,candes_TIT1,Eldar_B1,Rauhut_J1}: \emph{uniform} and \emph{nonuniform} recovery guarantees.  A \emph{uniform} recovery guarantee ensures  that for a given draw of the random projection matrix, all possible $k$-sparse signals are recovered with high probability.  On the other hand, \emph{nonuniform} recovery guarantee provides the conditions under which a given $k$-sparse signal (but not \emph{any} $k$-sparse signal as considered  in uniform recovery) can be reconstructed with a given draw of a random measurement matrix. Thus, \emph{uniform}  recovery focuses on the worst case recovery  guarantees  while \emph{nonuniform} recovery captures the typical recovery behavior of the measurement matrix.

  In this paper, the  goal is to enhance our understanding of  recovering  a given sparse signal  with  sparse random matrices  in the presence of channel fading.  More specifically, we provide lower bounds on the number of measurements that should be collected by the fusion center in order to achieve \emph{nonuniform} recovery guarantees with $l_1$ norm minimization based recovery with  independent (not necessarily identical) channel fading. With sparse random projections, the nodes transmit their observations with a certain probability. We further discuss how to design probabilities of transmissions by each node (equivalently the sparsity  parameter of the random projection matrix)    based on the channel fading statistics so that  the number of measurements required for signal recovery  at the fusion center is minimized.



While the authors in  \cite{Yang_TSP1} consider a similar problem of WCS, our analysis is different from that in  \cite{Yang_TSP1} in several ways.  In this paper,  we derive  \emph{nonuniform} recovery guarantees which require different derivations (not based on RIP) and provide   better recovery results  compared to  \emph{uniform} recovery as considered  in  \cite{Yang_TSP1}. {It is noted  that, the RIP
measure is  defined with respect to the worst-possible performance. Eventhough RIP  analysis adopts a probabilistic point of
view, the subsequent  results tend to be overly restrictive, leading
to a wide gap between theoretical predictions and actual performance
\cite{candes_TIT1}. With a given signal of interest, one can obtain  stronger  results. To that end, nonuniform recovery guarantees, as considered in this paper are  able to capture the typical recovery behavior of the projection matrix leading to  stronger results}.   We assume    envelope detection at the fusion center which  is employed in practice  in many sensor networks.  { More specifically, we assume that the channel phase is corrected to ensure phase coherence  which is a widely used  assumption in the  sensor network literature. As discussed in \cite{Cohen_TIT_13,Chen_TSP_07}, this can be achieved by  transmitting a
pilot signal by the fusion center  before the sensor transmissions to
estimate the channel phase}.  Further, the nonzero elements of the sparse matrices are assumed to be Gaussian.  Thus, the statistics of the low dimensional linear operator that relates the input and the output at the fusion center are  different from that in \cite{Yang_TSP1}. In particular, with the model considered in this paper,  the elements of the random projection matrix after taking channel fading into account reduce to independent but nonidentical  sub-exponential random variables. To the best of our knowledge,  \emph{nonuniform} recovery of a given sparse signal with nonidentical  sub-exponential (or heavy-tailed)  random matrices has not been well investigated in the literature.  Thus, the  analysis in this paper  further enhances our understanding on sparse recovery with sub-exponential random matrices in general.   Further, we  show that  the  number of measurements required to reconstruct  a given sparse signal can be reduced by designing  probabilities of transmission at each node  based on fading channel statistics.   In addition, our  results are in general not asymptotic while the results in \cite{Yang_TSP1} are asymptotic in nature.

Our main results are summarized below. In the presence of independent channel fading with Rayleigh distribution, we show  that the nodes  should transmit with a probability that is inversely proportional to the   fading channel statistics (channel power)   in order to reduce the number of measurements collected at the fusion center in recovering a given sparse signal. With this design of probabilities of transmissions, the number of measurements required to recover a given sparse signal  with sparsity index $k$  scales as $\left(\sqrt{\frac{\nu_{\max}^2}{\nu_{\min}^2}} k\log N\right)$ where $2\nu_{\max}^2$ and $2\nu_{\min}^2$ are the largest and smallest average mean power coefficients of Rayleigh fading channels, and  $N$ is the number of nodes in the network (which is assumed to be the same as the dimension of the sparse signal).   This says, by controlling the probability of transmission based on fading channel statistics, the impact of  inhomogeneity of the elements of the measurement matrix on signal recovery can be reduced leading to better recovery guarantees. In the special case where the fading channels are assumed to be identical and all the nodes transmit with the same probability  (say $0<\gamma\leq 1$), we show that $\mathcal O\left(\frac{k}{\sqrt{\gamma}}\log N\right)$ MAC transmissions are sufficient to recover a given sparse signal.
We further, provide detailed analysis on  recovery guarantees of a given sparse signal  with any random projection matrix where the elements are sub-exponential with a given sub-exponential norm.

The rest of the paper is organized as follows. In Section \ref{motivation}, the problem formulation is given. Recovery guarantees of a given sparse signal under independent channel fading are provided  in Section \ref{recovery}. We  discuss how to design probabilities of transmission based on channel fading statistics. Further, the results are specified when the  fading  channels are identical.  In Section \ref{nonuniform_anysubexp}, the conditions under which a given sparse signal can be recovered with any sub-exponential random matrix are discussed.   Numerical results are presented in Section \ref{numerical} and concluding remarks are given in Section \ref{conclusion}.

\subsection{Notation}
The following notation is used throughout the paper. Lower case boldface letters, e.g.,  $\mathbf x$ are used to denote vectors and the $j$-th element of $\mathbf x$ is denoted by $\mathbf x(j)$. Lower case letters are used to denote scalars, e.g.,  $x$.
Both upper case boldface letters and boldface symbols  are used to denote matrices, e.g.,  $\mathbf A$, $\boldsymbol\Phi$.  The notations, $\mathbf A_i$, $\mathbf a_i$ and $\mathbf A_{ij}$ are used to denote the $i$-th row, $i$-th column  and the $(i,j)$-th element of the matrix $\mathbf A$, respectively.  The transpose of a matrix or a vector is denoted by $(.)^T$ and $(\mathbf A)^{\dagger}$ denotes the Moore-Penrose pseudo inverse  of $\mathbf A$. The notation $\otimes$ denotes the outer product of two vectors.   Upper case letters with calligraphic font, e.g.,  $\mathcal S$, are used to denote sets.   The $l_p$ norm of a vector $\mathbf x$ is denoted by $||\mathbf x||_p$. The spectral norm of a matrix $\mathbf A$ is denoted by $||\mathbf A||$. We use the notation $|.|$ to denote the absolute value of a scalar, as well as the cardinality of a set. We use $\mathbf I_N$ to denote the identity matrix of dimension $N$ (we avoid using subscript when there is no ambiguity). A diagonal matrix in which the main diagonal consists of the vector $\mathbf x$ is denoted by $\mathrm{diag}(\mathbf x)$.  By $s_{\min}(\mathbf A)$ and $s_{\max}(\mathbf A)$, we denote the minimum and maximum singular values, respectively,  of the matrix $\mathbf A$. The notation $\mathrm{Rayleigh}(\sigma)$ denotes that a random variable $x$ has a Rayleigh distribution with the probability density function (pdf) $f(x) = \frac{x}{\sigma^2} e^{-\frac{x^2}{2\sigma^2}} $ for $x \geq 0$. The notation $x \sim \mathcal N(\mu, \sigma^2)$ denotes that the random variable $x$ is distributed as Gaussian with the pdf $f(x) =  \frac{1}{\sqrt{2\pi\sigma^2}} e^{-\frac{(x-\mu)^2}{2\sigma^2}}$.

\section{Problem Formulation}\label{motivation}
\subsection{Observation model}\label{Sec_obs}
Consider a distributed sensor network  measuring compressible (sparse) data using $N$ number of nodes. The observation collected at the $i$-th  sensor node is denoted by $x_i$ for $i=0,1,\cdots, N-1$. Let $\mathbf x=[x_0, \cdots, x_{N-1}]^T$  be the vector containing all the measurements of sensors. Sparsity is a common characteristic  observed with the data collected in  sensor networks. The sparsity may appear as an inherent property of the  signal being observed by multiple sensors, e.g., most acoustic data has a sparse representation in Fourier  domain. On the other hand, not all the observations collected at  nodes are informative; for example, the sensors located far away from the phenomenon being observed may contain poor observations making the observation vector $\mathbf x$ sparse in the canonical basis. In a general framework, assume that the signal $\mathbf x$ is sparse in some basis $\boldsymbol\Phi$. One of the fundamental tasks in many sensor networking applications is to reconstruct the  signal observed by nodes at a distant  fusion center. Due to inherent resource constraints in sensor networks, it is desirable that sensors use only small amount of energy and low   bandwidth while forwarding information to the fusion center.   With recent advances in the  theory of CS, WCS with random measurement matrices is becoming attractive. A compressed version of $\mathbf x$ can be transmitted to a fusion sensor exploiting coherent transmission schemes developed for sensor networks \cite{haupt4}.

Consider that  the $j$-th   sensor multiplies its observation during the $i$-th transmission  by $\mathbf A_{ij}$  which is a scalar (to be defined later). All the nodes transmit their scaled  observations coherently using  $M$  (time or frequency) slots.  { In this paper, we consider  the amplify-and-forward (AF) approach for sensor transmissions. It is noted that  a  digital approach can be used where
we digitize the observation into bits, possibly apply channel
coding, and then use digital modulation schemes to transmit the
data, for example, as considered in \cite{Bajwa_TIT_07}. However, as shown in  \cite{Gastpar_TIT_03}  for a single Gaussian source
with an AWGN channel, the AF approach is optimal. Analog transmission
schemes over MAC for detection and estimation using WSN
have been widely  investigated, for example, in \cite{Gastpar_TIT_08,mergen1,Marano_TSP_07}. Thus, we restrict our analysis in this paper to analog  transmission, while digital modulated signals will be considered in a future work}.  We further  assume that the channels between the sensors and the fusion center undergo flat fading. We further assume phase coherent reception, thus the effect of fading is  reflected as a scalar multiplication.
The received signal at the fusion center with the $i$-th MAC transmission is given by,
\begin{eqnarray}
y_i = \sum_{j=0}^{N-1} h_{ij}\mathbf A_{ij}x_j + v_i\label{obs_1}
\end{eqnarray}
for $i=0,1,\cdots, M-1$ where $h_{ij}$ is the channel coefficient for  the channel  between the $j$-th sensor and the fusion center during  the $i$-th transmission and $v_i$ is the additive noise with mean zero and variance $\sigma_v^2$.

Due to energy constraints in sensor networks, we consider  a scenario where not all  the nodes transmit during each MAC transmission. To achieve this, $\mathbf A_{ij}$ is selected as:
\begin{eqnarray}
\mathbf A_{ij}= \left\{
\begin{array}{cc}
a_{ij} & \mathrm{with}~\mathrm{prob}~ \gamma_{j}\\
0 & \mathrm{with}~\mathrm{prob}~ 1-  \gamma_{j}
\end{array}\right.\label{A_ij}
\end{eqnarray}
where $a_{ij} \sim \mathcal N(0,\sigma_{j}^2)$ and $0 < \gamma_j \leq 1$ is the probability of transmission of the $j$-th node.  { The  average power  used by the $j$-th sensor  during  the $i$-th MAC transmission is  $\mathbb E\{\mathbf A_{ij}^2\} =\gamma_{j}  \sigma_{j}^2$ which is assumed to be less than $  E_{j}$ where $ E_{j}$ is determined based on the available energy  at  the $j$-th node.  We  assume that the $j$-th node uses the same transmit  power  on an average  during all MAC transmissions.}   Let $\mathbf A$ be  a   $M\times N$ matrix in which $(i,j)$-th element is given by  $\mathbf A_{ij}$ as in (\ref{A_ij}). Further, let $\mathbf H$ be a $M\times N$ matrix in which $(i,j)$-th element is given by $h_{ij}$.
With  vector-matrix notation,  (\ref{obs_1}) can be written as,
\begin{eqnarray}
\mathbf y = \mathbf B \mathbf x + \mathbf v \label{obs_2}
\end{eqnarray}
where $\mathbf B = \mathbf H \odot \mathbf A$,  $\odot$ is the Hadamard (element-wise) product, $\mathbf y = [y_0, \cdots, y_{M-1}]^T$ and $\mathbf v=[v_0, \cdots, v_{M-1}]^T$. Let $\boldsymbol\gamma = [\gamma_0, \cdots, \gamma_{N-1}]^T$. The vector  $\boldsymbol\gamma$ is used to refer to the measurement sparsity of the matrix  $\mathbf B$ or equivalently the probabilities of transmission of  all the nodes. The goal is to recover $\mathbf x$ based on (\ref{obs_2}).  One of the widely used approaches for sparse recovery is to solve the following optimization problem \cite{Eldar_B1}:
  \begin{eqnarray}
\underset{\mathbf x}{\min} ||\mathbf x||_{1} ~ \mathrm{such }~ \mathrm{that} ~ \mathbf y = \mathbf B \mathbf x \label{l_1}
\end{eqnarray}
with no noise, or
\begin{eqnarray}
\underset{\mathbf x}{\min}  ||\mathbf x||_{1}  ~ \mathrm{such}~ \mathrm{that} ~ ||\mathbf y - \mathbf B \mathbf x||_2 \leq \epsilon_v\label{l_12}
\end{eqnarray}
with noise where $\epsilon_v$ bounds   the size of the noise term $\mathbf v$.

 It is noted that, when $\gamma_{j} = 1$ and $\sigma_{j}^2 = \sigma_{a}^2$ for all $j$,  the elements of $\mathbf A$  are  independent and identically distributed (iid) Gaussian with mean zero and variance $\sigma_a^2$. Then, if we further  assume AWGN channels so that $\mathbf B = \mathbf A$,  $\mathbf B$ is a random matrix with iid Gaussian random variables.  Sparse signal  recovery with iid Gaussian random matrices  has been  extensively studied \cite{Eldar_B1,candes_TIT1}.
Under fading, the matrix $\mathbf A$ is multiplied (element-wise)  by another random matrix $\mathbf H$ which has  independent and nonidentical elements. Thus, the recovery capability of (\ref{l_1}) (or (\ref{l_12})) depends  on the properties of the matrix $\mathbf B = \mathbf H \odot \mathbf A$. In this paper, we assume that the fading coefficients $h_{ij}$ are independent Rayleigh random variables with $h_{ij} \sim \mathrm{Rayleigh}(\nu_{j})$ for $i=0, \cdots, M-1$ and $j=0, 1, \cdots, N-1$ where $\mathbb E\{h_{ij}^2\} = 2\nu_{j}^2$ is assumed to be different in general for the channels between different sensors and the fusion center. The  goal is to obtain recovery guarantees of  a given $\mathbf x$ based on (\ref{l_1}) under the above discussed  statistics for $\mathbf A$ and $\mathbf H$.


First, it is important to observe the statistical properties of the matrix $\mathbf B$.
\subsection{Statistics of $\mathbf B$}
The $(i,j)$-th  element of $\mathbf B$ is given by,
\begin{eqnarray}
\mathbf B_{ij}= \left\{
\begin{array}{cc}
h_{ij}a_{ij} & \mathrm{with}~\mathrm{prob}~ \gamma_{j}\\
0 &\mathrm{ with}~\mathrm{prob}~ 1-  \gamma_{j}\label{B_ij}
\end{array}\right.
\end{eqnarray}
for $i=0,\cdots, M-1$ and $j=0,\cdots, N-1$.
Since the elements of $\mathbf A$ and $\mathbf H$ are assumed to be independent, the elements of $\mathbf B$ are also independent (but not identical in general).
\begin{proposition}\label{prop_1}
Let $w=h_{ij}a_{ij}$ where  $a_{ij} \sim \mathcal N(0, \sigma_a^2)$ and $h_{ij} \sim \mathrm {Rayleigh}(\nu_h)$.  Then the pdf of $w$ is doubly exponential (Laplacian) which is given by,
\begin{eqnarray*}
f(w) = \frac{1}{2\bar{\sigma}}  e^{- {\frac{1}{\bar{\sigma}}} |w|}.
\end{eqnarray*}
where $\bar{\sigma} = \nu_h\sigma_a$.
\end{proposition}

\begin{proof}
See Appendix A.
\end{proof}

Taking  $\bar{\sigma}_j = \sigma_{j} \nu_{j}$, the pdf of $u= \mathbf B_{ij}$ in (\ref{B_ij}) can be written as
\begin{eqnarray}
f(u) = \gamma_{j} \frac{1}{2 \bar{\sigma}_j } e^{-\frac{|u|}{\bar{\sigma}_j}} + (1-\gamma_{j}) \delta(u)\label{f_u}
\end{eqnarray}
where  $\delta(.)$ is the Dirac delta function. It can be easily proved  that
\begin{eqnarray*}
Pr(|u| > t) = \left\{
\begin{array}{ccc}
1~& for~ t=0\\
\gamma_{j} e^{-\frac{t}{\bar{\sigma}_j}}~ & for~ t>0
\end{array}\right. .
\end{eqnarray*}
Thus, we can find a constant $K_1 > 0$ such that,
\begin{eqnarray*}
Pr(|u| > t) \leq e^{1-t/K_1}
\end{eqnarray*}
for all $t\geq 0$. Thus, $u$ is a sub-exponential random variable \cite{Vershynin_book1}. In other words, the elements of $\mathbf B$ are independent (but not identical in general)  sub-exponential random variables. While there is a  substantial amount of work in the literature that addresses the problem of  sparse signal recovery with Gaussian and sub-Gaussian random matrices, very little is known  with   random matrices with sub-exponential (or heavy tailed) elements. In the following, we obtain nonuniform recovery guarantees for (\ref{l_1}) when the elements of $\mathbf B$ have a pdf as given in (\ref{f_u}) and simplify the results when the matrix $\mathbf B$ is isotropic. We further provide recovery guarantees for  general nonidentical sub-exponential random matrices.

\section{Nonuniform  Recovery Guarantees with Independent Channel Fading}\label{recovery}
We present  the following statistical results which are helpful in deriving recovery conditions.

\begin{definition}[Isotropic random vectors \cite{candes_TIT1}]
A random vector $\mathbf x \in \mathbb R^N$ is called isotropic if $\mathbb E\{\mathbf x \mathbf x^T\} = \mathbf I_N$.
\end{definition}
The row vectors of the matrix  $\mathbf B$ are in general not isotropic. However, the column vectors of $\mathbf B$ with appropriate normalization  become isotropic. In  the  special case where $\gamma_j =\gamma$ and $\bar\sigma_j =  \bar\sigma$, both row and columns vectors of the normalized matrix $\frac{1}{\sqrt{2 \gamma \bar\sigma^2}} \mathbf B$  are isotropic.

\begin{proposition}[mgf $u$]\label{prop_mgf_sub_exp}
Let $u$ be a random variable with pdf $f(u)$ where $f(u)$ is given in (\ref{f_u}). Then for $|t| \leq \frac{1}{\tilde{\eta}_{\max} }$ where $\tilde{\eta}_{\max}  = \underset{j}{\max}\{\bar\sigma_j\}$, we have
\begin{eqnarray*}
\mathbb E \{e^{tu}\} \leq  e^{\eta_{\max} t^2}
\end{eqnarray*}
where $\eta_{\max} = \underset{j}{\max}\{\gamma_j \bar\sigma_j^2\}$.
\end{proposition}

\begin{proof}
See Appendix B.
\end{proof}

Next, we provide a Bernstein-type inequality to bound the weighted sum of independent but nonidentical  random variables with pdf as given in (\ref{f_u}). It is noted that, a similar bound is derived in \cite{Vershynin_book1} for general sub-exponential random variables which are characterized by the sub-exponential norm. The following results are the same as  those in Proposition $5.16$ of  \cite{Vershynin_book1} only when $\gamma_j = 1$ for all $j$.

\begin{proposition}[Bernstein-type inequality ]\label{prop_bern_inequ}
Let $u_0, \cdots, u_{N-1}$ be $N$ independent random variables where the pdf of $u_j$ is as given in (\ref{f_u}) for $j=0,\cdots, N-1$. Then for every $\boldsymbol\alpha = (\alpha_0, \cdots, \alpha_{N-1})\in \mathbb R^N $ and every $t > 0$, we have,
\begin{eqnarray}
Pr\left(|\sum_{i=0}^{N-1} \alpha_i u_i| \geq t\right) \leq 2 e^{- \min\left(\frac{t^2}{4\eta_{\max} ||\boldsymbol\alpha||_2^2}, \frac{t}{2 \tilde\eta_{\max}||\boldsymbol\alpha||_{\infty}}\right)}\label{bern_noniid}
\end{eqnarray}
where $\eta_{\max} = \underset{j}{\max} \{\gamma_j \bar\sigma_j^2\}$ and $\tilde\eta_{\max} = \underset{j}{\max}\{\bar\sigma_j\}$ as defined before.
\end{proposition}

\begin{proof}
See Appendix C.
\end{proof}

\subsection{Nonuniform recovery guarantees in the presence of independent  fading channels}
In the following, we present our main results on  recovery of a given  $\mathbf x$ based on (\ref{l_1}). Before that,  we introduce additional  notation.  Let $\mathcal U = \{0,1,\cdots, N-1\}$ and $\mathcal S:= \mathrm{supp}(\mathbf x)= \{i: \mathbf x(i) \neq 0, i=0, 1,\cdots, N-1\}$ where $\mathbf x(i)$ is the $i$-th element of $\mathbf x$. For a $k$-sparse vector $\mathbf x$,  we have $|\mathcal S|=k$.  Further, by $\mathbf B_S$, we denote the sub-matrix  of $\mathbf B$ that contains columns of $\mathbf B$ corresponding to the indices in $\mathcal S$ and $\mathbf x^S$ is a $k \times 1$ vector which contains the elements of $\mathbf x$ corresponding to indices in $\mathcal S$. Further, let $\bar{\boldsymbol\sigma} = [\bar{\sigma}_{0}, \cdots, \bar{\sigma}_{N-1}]^T$ and $\boldsymbol\nu = [\nu_0, \cdots, \nu_{N-1}]^T$.

To ensure recovery of a given signal $\mathbf x$ via (\ref{l_1}), it is sufficient to show that \cite{Ayaz_Tech1, Fuchs_TIT1,Rauhut_J1},
\begin{eqnarray*}
|\langle (\mathbf B_S)^{\dagger} \mathbf b_l, \mathrm{sgn}(\mathbf x^S)\rangle| < 1 ~ \mathrm{for}~\mathrm{all}~ l\in \mathcal U \setminus \mathcal S
\end{eqnarray*}
where $(\mathbf B_S)^{\dagger} = (\mathbf B_S^T \mathbf B_S)^{-1} \mathbf B_S^T$ is the Moore-Penrose pseudo inverse of $\mathbf B_S$ and $\mathrm{sgn}(\mathbf x)$ is the sign vector having entries
\begin{eqnarray*}
\mathrm{sgn}(\mathbf x)_j := \left\{
\begin{array}{ccc}
\frac{\mathbf x_j}{|\mathbf x_j|} ~ & \mathrm{if} ~ \mathbf x(j) \neq 0,\\
0, ~ & \mathrm{otherwise},
\end{array}\right.
\mathrm{for} ~\mathrm{all} ~j\in \mathcal U.
\end{eqnarray*}

\begin{theorem}\label{theorem_1}
Let $\mathcal S \subset \mathcal U$ with $|\mathcal S| = k$. Further, let the elements of $\mathbf B$ be given  as in (\ref{B_ij}),  $\eta_{\max}(\boldsymbol\gamma, \bar{\boldsymbol\sigma}) = \underset{0\leq j\leq N-1}{\max} (\gamma_j \bar\sigma_j^2)$, $\eta_{\min}(\boldsymbol\gamma, \bar{\boldsymbol\sigma}) = \underset{0\leq j\leq N-1}{\min} (\gamma_j \bar\sigma_j^2)$ and $\tilde\eta_{\max} ( \bar{\boldsymbol\sigma})= \underset{0\leq j\leq N-1}{\max}\{\bar\sigma_j \} $.   Define $R$ such that
\begin{eqnarray*}
 \frac{||\mathbf b_S||_{\infty}}{||\mathbf b_S||_2} \leq R
  \end{eqnarray*}
  almost surely where $0 < R \leq 1$ and $\mathbf b_S = (\mathbf B_{\mathcal S}^{\dagger})^T \mathrm{sgn}(\mathbf x^S) $.  Then, for $0< \epsilon, \epsilon' < 1$, $\mathbf x$ is the unique solution  to (\ref{l_1}) with probability exceeding $1-\max(\epsilon,\epsilon')$ if the following condition is  satisfied:
\begin{eqnarray}
M \geq \max\{M_1, M_2\} \label{M_noniid}
\end{eqnarray}
where $M_1$ and $M_2$ are given in (\ref{M_1}) and (\ref{M_2}) respectively,
\begin{figure*}
\begin{eqnarray}
M_1 &=& \frac{\eta_{\max}(\boldsymbol\gamma, \bar{\boldsymbol\sigma})}{\eta_{\min}(\boldsymbol\gamma, \bar{\boldsymbol\sigma})} 2k \left(\sqrt{\frac{\eta_{\max}(\boldsymbol\gamma, \bar{\boldsymbol\sigma})}{\eta_{\min}(\boldsymbol\gamma, \bar{\boldsymbol\sigma})}}\sqrt{\log(2N/\epsilon)} + \sqrt{\frac{\log(k/\epsilon')}{2c'}}\right)^2\label{M_1}\\
M_2& =& \frac{\eta_{\max}(\boldsymbol\gamma, \bar{\boldsymbol\sigma})}{\eta_{\min}(\boldsymbol\gamma, \bar{\boldsymbol\sigma})} 2k \left(\frac{\tilde\eta_{\max}( \bar{\boldsymbol\sigma})}{\sqrt{\eta_{\min}(\boldsymbol\gamma, \bar{\boldsymbol\sigma})}}R{\log(2N/\epsilon)} + \sqrt{\frac{\log(k/\epsilon')}{2c'}}\right)^2\label{M_2}
 \end{eqnarray}
 \end{figure*}
and  $c'$ is an absolute constant.
\end{theorem}

\begin{proof}
See Appendix D.
\end{proof}

From Theorem \ref{theorem_1}, it is observed that the ratio between peak and total energy of $\mathbf b_S$, $R$,  plays an important  role in deciding the minimum number of MAC transmissions needed to recover $\mathbf x$ with a given support $\mathcal S$. As shown in Appendix E, when the elements of $\mathbf B$ are distributed according to (\ref{f_u}) we can take
\begin{eqnarray*}
R = \mathcal O\left(\sqrt{\frac{k}{2M}}\right).
\end{eqnarray*}
Then, the dominant part of $M_2$ in (\ref{M_2}) scales as
\begin{eqnarray}
\mathcal O\left(\sqrt{\frac{\eta_{\max}(\boldsymbol\gamma, \bar{\boldsymbol\sigma})\tilde\eta^2_{\max}(\bar{\boldsymbol\sigma})}{\eta^2_{\min}(\boldsymbol\gamma, \bar{\boldsymbol\sigma})}} k {\log(2N / \epsilon)}\right) \label{M_2_dominant}
 \end{eqnarray} while the dominant term of $M_1$ scales as
 \begin{eqnarray}
 \mathcal O\left(\frac{\eta^2_{\max}(\boldsymbol\gamma, \bar{\boldsymbol\sigma})}{\eta^2_{\min}(\boldsymbol\gamma, \bar{\boldsymbol\sigma})} k {\log(2N/\epsilon)}\right).\label{M_1_dominant}
 \end{eqnarray} Thus, when
\begin{eqnarray}
\underset{j}{\min} \{\gamma_j \bar\sigma_j^2\} \leq \underset{j}{\max} \{\gamma_j \bar\sigma_j^2 \} \frac{\sqrt{\underset{j}{\max} \{\gamma_j \bar\sigma_j^2\}}}{\underset{j}{\max}\{ \bar\sigma_j\}}\label{M_1DominatesM_2}
\end{eqnarray}
$M_1$ dominates $M_2$ (and vice versa).

\subsection{Probabilities of transmission and channel fading statistics }
From (\ref{M_2_dominant}) and (\ref{M_1_dominant}), it is seen that the number of MAC  transmissions  required for reliable signal recovery depends  on the probabilities of transmission  $\boldsymbol\gamma$, and the quality of the fading channels  $\boldsymbol\nu$.
Since the designer has the control on $\boldsymbol\gamma$,  we discuss how to  design   $\boldsymbol\gamma$ as a function of $\boldsymbol\nu$  so that $M_1$ and $M_2$ become minimum with respect to $\boldsymbol\gamma$.  Since $\eta_{\max} \geq \eta_{\min}$, for $M_1$ and $M_2$ to be minimum,
it is desired to have the gap between $\eta_{\max} $ and $\eta_{\min}$ minimum. When $\eta_{\max} = \eta_{\min}$, it is easily seen that $M_2$ dominates $M_1$, thus, $M= M_2$.   Let us assume that  the maximum available energy at each node for given transmission is the same so that $\sigma_j^2 = \sigma_a^2$ and $E_j=E$   for all $j$.  Then  the probability of transmission at each node should satisfy the following condition:
\begin{eqnarray}
\gamma_j \leq \min \left\{ 1, \frac{ E}{ \sigma_a^2}\right\} = \bar\gamma\label{gamma_const}
\end{eqnarray}
for $j=0, \cdots, N-1$.

When $\sigma_j^2 = \sigma_a^2$ for all $j$, the term that depends on $\boldsymbol\gamma$ in $M_2$ in (\ref{M_2_dominant}) can be expressed as,
\begin{eqnarray}
\Psi(\boldsymbol\gamma) = \sqrt{\frac{\eta_{\max}(\boldsymbol\gamma, \bar{\boldsymbol\sigma})\tilde\eta^2_{\max}(\bar{\boldsymbol\sigma})}{\eta^2_{\min}(\boldsymbol\gamma, \bar{\boldsymbol\sigma})}} =\sqrt{ \frac{\underset{j}{\max} \{\gamma_j \nu_j^2\} \underset{j}{\max} \{ \nu_j^2\}}{\left(\underset{j}{\min} \{\gamma_j \nu_j^2\}\right)^2}}.\label{Psi_gamma}
\end{eqnarray}
Since $\gamma_j \leq 1$ for all $j$, we have $\Psi(\boldsymbol\gamma) \geq 1$ and the { equality (of $\Psi(\boldsymbol\gamma) \geq 1$ where $\Psi(\boldsymbol\gamma) $ is given in (\ref{Psi_gamma}))} holds only if $\gamma_j =1$ for all $j$ and the channels are identical so that $\nu_0^2=\nu_1^2= \cdots, \nu_{N-1}^2$. The goal is to find $\boldsymbol\gamma$ so that (\ref{Psi_gamma}) is minimized under the constraint (\ref{gamma_const}). It is noted that  $\Psi(\boldsymbol\gamma)$ is minimum with respect to $\boldsymbol\gamma$ when $\underset{j}{\max} \{\gamma_j \nu_j^2\}  =\underset{j}{\min} \{\gamma_j \nu_j^2\}$.  Without loss of generality, we sort   $\nu_j$'s in ascending order so that $\nu_1 \leq \nu_2 \leq \cdots, \nu_N$.  To achieve $\underset{j}{\max} \{\gamma_j \nu_j^2\}  =\underset{j}{\min} \{\gamma_j \nu_j^2\}$, we select $\boldsymbol\gamma$ so that the largest $\gamma_j$  is assigned to the node indexed by $0$ while the smallest $\gamma_j$ is assigned to the node indexed by $N-1$ for  $j=0,\cdots, N-1$. More specifically, let $\gamma_j = \frac{d_0}{\nu_j}$ for $j=0,\cdots, N-1$ where $d_0$ is a constant. This leads to
\begin{eqnarray}
\Psi(\boldsymbol\gamma) = \sqrt{\frac{\nu_{N-1}^2}{d_0}}. \label{psi_gamma_2}
\end{eqnarray}
With the  constraint for $\gamma_j$ in (\ref{gamma_const}), we further have
\begin{eqnarray*}
d_0 \leq \bar\gamma \nu_0^2.
\end{eqnarray*}
Then,   $\Psi(\boldsymbol\gamma)$ in (\ref{psi_gamma_2}) is  minimum when $d_0 = \bar\gamma \nu_0^2$. Thus, the probabilities of transmission  which minimize  $\Psi(\boldsymbol\gamma)$ are given by
\begin{eqnarray}
\gamma_j^{\mathrm{opt}} = \bar\gamma \frac{\nu_0^2}{\nu_j^2}\label{gamma_opt}
\end{eqnarray}
for $j=0,\cdots, N-1$ and the minimum value of $\Psi(\boldsymbol\gamma)$ is
\begin{eqnarray*}
\Psi(\boldsymbol\gamma^{\mathrm{opt}}) = \sqrt{\frac{\nu_{N-1}^2}{\bar\gamma \nu_0^2}}.
\end{eqnarray*}

With this design of $\boldsymbol\gamma$, the number of MAC transmissions required for reliable  signal recovery at  the fusion center scales as
\begin{eqnarray}
M = \mathcal O\left( \sqrt{\frac{C_1(\boldsymbol\nu)}{\bar\gamma}} k {\log(2N / \epsilon)} \right) \label{M_opt_gamma}
\end{eqnarray}
where $C_1(\boldsymbol\nu) = \frac{\underset{j}{\max}\{\nu_{j}^2\}}{\underset{j}{\min}\{\nu_{j}^2\}}$
Thus, the impact of inhomogeneous channel fading with the optimal design of $\boldsymbol\gamma$ on $M$  appears as the ratio between $\underset{j}{\max}\{\nu_{j}\}$ and $\underset{j}{\min}\{\nu_{j}\}$.

\subsection{Nonuniform recovery when $\mathbf B$ is dense}\label{Dense_B}
Here we study the  special case where  $\gamma_j= 1$ for $j =0, \cdots N-1$ so that $\mathbf B$ is a dense matrix.  In this case,  $M_1$ in (\ref{M_1_dominant})  dominates  $M_2$ in  (\ref{M_2_dominant}). Thus, $M=M_1$ and we have, \begin{eqnarray}
M = \mathcal O\left(C_2(\boldsymbol\nu) k {\log(2N/\epsilon)}\right)\label{M_gamma1}
\end{eqnarray}
where $C_2(\boldsymbol\nu) = \left(\frac{\underset{j}{\max}\{\nu_{j}^2\}}{\underset{j}{\min}\{\nu_{j}^2\}}\right)^2$ and we assume $\sigma_j^2 = \sigma_a^2$ for all $j$.
Note that in this case $\bar\gamma=1$. From (\ref{M_opt_gamma}) and (\ref{M_gamma1}), it is seen that the scaling of $M$ when $\gamma=1$  is greater than that is with a sparse matrix with   properly designed probabilities transmission since $C_2(\boldsymbol\nu)\geq C_1(\boldsymbol\nu)$. This implies that, with nonidentical fading channels, it is beneficial to use sparse random projections with transmission probabilities matched  to fading statistics as in (\ref{gamma_opt})  compared to the use of dense matrices in order to reduce the total number of MAC transmissions. While (\ref{M_gamma1}) provides a scaling, the exact  $M$ required for reliable sparse signal recovery  is illustrated in numerical results section (Fig. \ref{fig_MSE_M_diffgamma_noniidfading}) for dense and sparse matrices with nonidentical channels.

As will be shown in the next section, when the matrix $\mathbf B$ has dense iid elements (so that  $\gamma_j= 1$ for $j =0, \cdots N-1$  and $\nu_0^2 = \cdots= \nu_{N-1}^2$),   we get  $M = \mathcal O (k\log(2N/\epsilon))$. From (\ref{M_opt_gamma}) and (\ref{M_gamma1}), it is seen that, the presence of non identical fading channels (the inhomogeneity) increases the required number of MAC transmissions by a  factor of $\sqrt{\frac{C_1(\boldsymbol\nu)}{\bar\gamma}}$ with sparse projections  and $C_2(\boldsymbol\nu)$ with dense projections, respectively,   compared to that required with   identical   channels. It is further worth mentioning   that we obtain dominant parts of $M_1$ and $M_2$ as in   (\ref{M_1_dominant}) and  (\ref{M_2_dominant})  using  lower bounds for (\ref{M_1}) and (\ref{M_2}), respectively. Thus, the impact of $C_1(\boldsymbol\nu)$ and $C_2(\boldsymbol\nu)$  on (\ref{M_opt_gamma}) and (\ref{M_gamma1}), respectively,  can  be scaled versions of them.

\subsection{Nonuniform recovery when $\mathbf B$ is isotropic}

Now consider the special case where $\gamma_j =\gamma$, $\bar{\sigma}_j = \bar{\sigma}=\sigma_a\nu_h$ for all $j$. Then,  the elements of $\mathbf B $ are iid random variables  and the columns and rows of the  scaled random matrix $\frac{1}{\sqrt{2 \gamma \bar\sigma^2}} \mathbf B$ are isotropic.  From Theorem \ref{theorem_1}, we have the following Corollary.
\begin{corollary}\label{corollary_2}
Assume  $\gamma_j = \gamma$ and $\bar{\sigma}_j = \bar{\sigma}$ for all $j$. Then when
 \begin{eqnarray}
 M=  \mathcal O\left(\frac{k}{\sqrt{\gamma}} \log(2N/\epsilon)\right)   \label{M_iid}
 \end{eqnarray} 
 $\mathbf x$ can be uniquely determined based on (\ref{l_1}) with high probability, where $0< \epsilon < 1$ is as defined in Theorem \ref{theorem_1}.
\end{corollary}

\begin{proof}
When $\gamma_j = \gamma$ and $\bar{\sigma}_j = \bar{\sigma}$ for all $j$, we have $\eta_{\min} = \eta_{\max} = \gamma \bar\sigma^2$ and $\tilde\eta_{\max} = \bar\sigma$. Then, the  scaling of $M_1$ in (\ref{M_1_dominant}) reduces to $\mathcal O\left(k \log (2N/\epsilon)\right)$ while  the scaling of $M_2$ in (\ref{M_2_dominant}) reduces to $\mathcal O\left(\frac{k}{\sqrt \gamma} \log (2N/\epsilon)\right)$0. Since $\gamma \leq 1$, $M_2$ dominates $M_1$.
\end{proof}

When $\gamma=1$, the matrix $\mathbf B$ is dense and  the elements are iid doubly   exponential. Then $\mathcal O (k\log(2N/\epsilon))$ measurements are sufficient for reliable recovery of $\mathbf x$. As $\gamma$ decreases, equivalently when the matrix $\mathbf B$ becomes more sparse, the minimum $M$ required for sparse signal recovery  increases.  In particular, when $\gamma < 1$,  the product $\gamma k $ plays an important role in determining $M$. It is noted that $\gamma k $ reflects the average number of nonzero coefficients  of $\mathbf x$ that align with the nonzero coefficients in each row of the sparse projection  matrix $\mathbf B$.  In Table \ref{table_M_gamma_k}, we summarize the scalings of $M$ required for recovery of $\mathbf x$ in  different regimes of $\gamma k$. In particular,
\begin{itemize}
\item when $\gamma k = \tau_0$ where $\tau_0$ is a constant, we have $\gamma \propto \frac{1}{k}$. Then, when  $k$ is sublinear with respect to $N$ so that $k = o(N)$,  $\mathcal O(k^{3/2} \log(N))$ measurements are sufficient for reliable recovery of given $\mathbf x$. It is noted that this scaling is only slightly greater than  $\mathcal O(k \log(N))$ which is the scaling   required for a dense matrix with iid elements.  This observation  is intuitive since, when $k=o(N)$, $\gamma$ is not very small and the matrix $\mathbf B$ is not 'very' sparse. On the other hand, when $k$ is linear with respect to $N$ so that $k = \Theta(N)$, and $\gamma \propto \frac{1}{k}$, $\mathcal O(N^{3/2} \log(N))$ measurements are required.
    \item when $\gamma k = \varepsilon  N$ with  $0<\varepsilon < \frac{k}{N}$, it is required to have $\mathcal O\left(\frac{k^{3/2}}{\sqrt {\varepsilon N}} \log(N)\right)$ measurements when $k = o(N)$. It is noted that, with this setting we have $\varepsilon<  \frac{k}{N}$ and $\varepsilon  \rightarrow 0$ as $N\rightarrow\infty$.   On the other hand, when $k = \Theta(N)$ and $\gamma k = \varepsilon N$ with  $0<\varepsilon < \frac{k}{N}$,   $M $ should be scaled as $\mathcal O\left(\frac{N}{\sqrt{\varepsilon}} \log(N)\right)$. With this setting $\varepsilon < \Theta (1)$, thus, $\mathcal O(N\log (N))$ measurements are needed  for reliable recovery of $\mathbf x$.
    \end{itemize}

\begin{table}
\caption{Minimum  M in the different regimes of $\gamma k$}
\centering
\centering
\begin{small}
\begin{tabular}{|l|l|l|}
  \hline
 $\gamma k$  & $M$ when $k=o (N)$ & $M$ when $k=\Theta(N)$ \\
  \hline
 $\gamma k = \tau_0$ & $\mathcal O \left(k^{3/2}\log(N)\right)$ & $\mathcal O(N^{3/2} \log(N))$  \\
     \hline
     $\gamma k = \varepsilon N$ & $\mathcal O\left(\frac{k^{3/2}}{\sqrt {\varepsilon N}} \log(N)\right)$ & $\mathcal O\left(\frac{N}{\sqrt{\varepsilon}} \log(N)\right)$\\
     $0 < \varepsilon < \frac{k}{N}$ & ~ & \\
     \hline
 \end{tabular}
\end{small}\label{table_M_gamma_k}
\end{table}

\subsubsection{Design of $\gamma$ under total network energy constraints}
For given $M$, the average energy  required by the network to achieve complete sparse signal recovery in the presence of iid fading channels  is given by,
 \begin{eqnarray*}
 E_{req} =  M \gamma N \sigma_a^2.
 \end{eqnarray*}
For complete signal recovery with probability at least $1-\epsilon$, we should have
 \begin{eqnarray*}
M \geq C_0 \frac{k}{\sqrt\gamma} \log(2N/\epsilon)
\end{eqnarray*}
where $C_0$ is a constant.
Then, we have,
 \begin{eqnarray}
 E_{req} \geq
 \sqrt{ \gamma} C_0 \sigma_a^2 k N \log(2N/ \epsilon).     \label{E_req}
 \end{eqnarray}
Assume that the network is subject to a total energy constraint so that we have to make sure,
\begin{eqnarray*}
 E_{req}\leq \bar E.
\end{eqnarray*}
Then, $\gamma$ should satisfy the following constraint:
\begin{eqnarray*}
\sqrt \gamma \leq \min\left\{1,\frac{\bar E}{C_0 \sigma_a^2 k N \log(2N/ \epsilon) }\right\}.
\end{eqnarray*}

\section{Nonuniform Recovery Guarantees with  General Sub-exponential Matrices}\label{nonuniform_anysubexp}
In  the following, we consider recovering $\mathbf x$ from  $\mathbf y =\mathbf B \mathbf x + \mathbf v$ when the elements of  $\mathbf B$  are general sub-exponential random variables and the rows of $\mathbf B$ are non-isotropic.

First,  let us define the sub-exponential norm of a sub-exponential random variable which will be helpful in the following analysis.
 \begin{definition}[sub-exponential norm \cite{Vershynin_book1}]
Let $x$ be a sub-exponential random variable. The sub-exponential norm of $x$, $||x||_{\psi_1}$, is defined by
\begin{eqnarray*}
||x||_{\psi_1} = \underset{p\geq 1}{\sup}~ \frac{1}{p} (\mathbb E \{|x|^p\})^{1/p}.
\end{eqnarray*}
\end{definition}
Further, let us assume that the each row of $\mathbf B$  has the same second moment matrix $\Sigma_B$.
Then,  we have the following Theorem.
\begin{theorem}\label{theorem_general_sub}
Let $\mathbf x$ be a $k$-sparse vector with the support set $\mathcal S$ and the matrix $\mathbf B$ contain independent sub-exponential random variables. Let $\rho_{\max}$ denotes the maximum sub-exponential norm  over all  the realizations. Further, assume that rows of $\mathbf B$, $\mathbf B_i$'s have  the same second moment matrix $\Sigma_B$ and $||\mathbf B_i||_2 \leq \sqrt{T_0}$ almost surely for all $i$. Let $\lambda_{\min}$ denote the minimum eigenvalue of $\Sigma_B^T \Sigma_B$.  Then, when the number of measurements
\begin{eqnarray}
M \geq \frac{1}{\lambda_{\min}} \left( \frac{\sqrt{k}}{\beta_1} + \sqrt{\frac{T_0 \log(k/\epsilon'_1)}{c_1'}}\right)^2\label{M_general_sub}
\end{eqnarray}
(\ref{l_1}) provides the unique solution for $\mathbf x$ with a given  support $\mathcal S$ with probability exceeding $1-\max(\epsilon_1, \epsilon'_1)$ where
\begin{eqnarray*}
\beta_1 =  \min\left( \frac{1}{\rho_{\max}}\sqrt{\frac{c_1}{\log(2N/\epsilon_1)}}, \frac{c_1}{\rho_{\max} R_1 \log(2N/\epsilon_1)}\right)
\end{eqnarray*}
and  $R_1$ is defined  such that $\frac{||\mathbf b_S||_{\infty}}{||\mathbf b_S||_2} \leq R_1$ almost surely,  and $\mathbf b_S = (\mathbf B_S^{\dagger})^T \mathrm{sgn}(\mathbf x^S)$ as defined before.
\end{theorem}
\begin{proof}
See Appendix F.
\end{proof}

The dominant part of $M$ in (\ref{M_general_sub}) scales as
\begin{eqnarray*}
M = \mathcal O \left(\frac{\rho_{\max}^2}{\lambda_{\min}} k \tilde{\beta}_1(N)\right)
\end{eqnarray*}
where
\begin{eqnarray*}
\tilde{\beta}_1(N) = \max\left\{\frac{\log(2N/\epsilon_1)}{c_1}, \frac{R_1^2 \log^2(2N/\epsilon_1)}{c_1^2}\right\}.
\end{eqnarray*}
Then, (\ref{l_1}) provides a unique solution for $\mathbf x$ with high probability if
\begin{eqnarray}
M = \left\{
\begin{array}{ccc}
\mathcal O\left(\frac{\rho_{\max}^2}{\lambda_{\min}} k R_1^2 \log^2(2N/\epsilon_1)\right)~& if~ R_1 \geq \mathcal O\left(\frac{1}{\sqrt{\log N}}\right)\\
\mathcal O\left(\frac{\rho_{\max}^2}{\lambda_{\min}} k \log(2N/\epsilon_1)\right)~& if~ R_1 \leq \mathcal O\left(\frac{1}{\sqrt{\log N}}\right)
\end{array}\right. .
\end{eqnarray}

Thus, it is observed that,  a threshold on $R_1$, the maximum peak-to-average energy of $\mathbf b_s$ over all $\mathcal S$ plays an important role in determining  the number of compressive measurements required for reliable sparse signal recovery with random matrices with general sub-exponential random variables.

\section{Numerical  Results}\label{numerical}
In this section, we provide numerical results to illustrate the performance of sparse signal recovery in the presence of fading. We consider that the nonzero entries of $\mathbf x$ are drawn from a uniform distribution in the range $[-20,-10] \cup [10, 20]$. { For numerical results, the primal-dual interior point method is used to  solve for $\mathbf x$  in  (\ref{l_1})  while (\ref{l_12})  is solved after converting to the second-order cone program  as presented in \cite{Candes_l1magic2005}. In Figures  \ref{fig_MSE_gamma_M}-\ref{fig_MSE_M_diffgamma_noniidfading}, the problem posed in  (\ref{l_1}) is considered where the noise power at the fusion center is assumed to be zero.}

\subsection{iid fading channels and identical  measurement  sparsity parameters}
First, we assume that $\sigma_j^2 = \sigma_a^2$, $\gamma_j =\gamma$ and $\nu_j^2=\nu_h^2$ for $j=0,\cdots, N-1$. The performance metric is taken as MSE which  is defined as,
\begin{eqnarray}
MSE = \mathbb E\left\{\frac{||\mathbf x - \hat{\mathbf x}||_2}{||\mathbf x||_2}\right\}
\end{eqnarray}
where $\hat{\mathbf x}$ is the estimated signal.

\begin{figure}
\centerline{\epsfig{figure=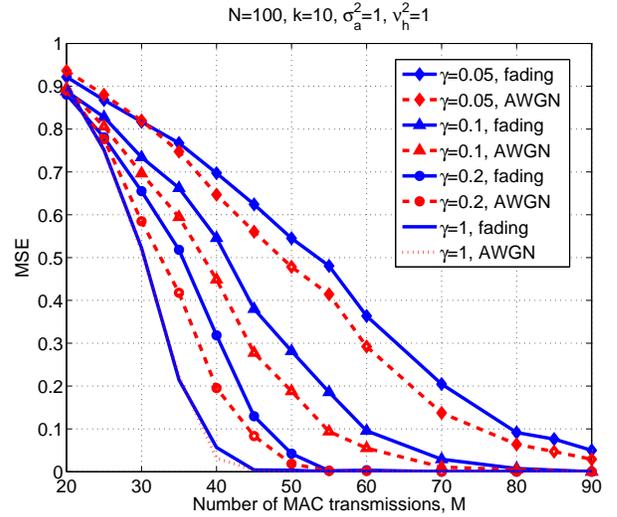,width=9.0cm}}
\caption{MSE vs number of MAC transmissions with iid fading channels}
  \label{fig_MSE_gamma_M}
\end{figure}

\begin{figure}
\centerline{\epsfig{figure=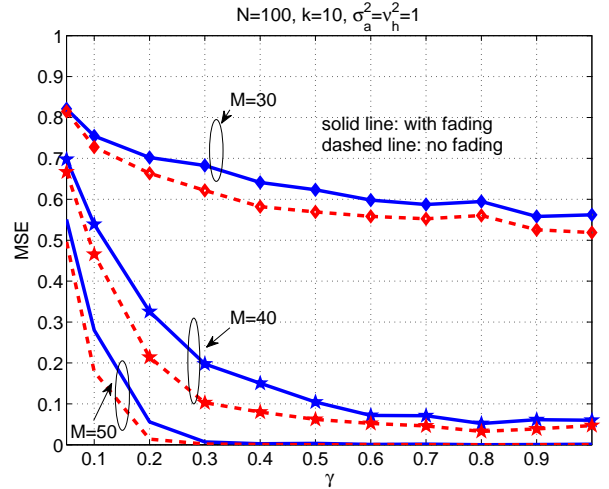,width=9.0cm}}
\caption{MSE vs measurement sparsity index $\gamma$ with iid fading channels}
  \label{fig_MSE_gamma}
\end{figure}
In Figs. \ref{fig_MSE_gamma_M} and \ref{fig_MSE_gamma},   the MSE vs number of MAC transmissions $M$ and the measurement sparsity index $\gamma$, respectively  is plotted for $N=100$, $k=10$, $\sigma_a^2 = 1$ and $\nu_h^2 =1 $. We further plot the performance in the absence of  fading; i.e. assuming AWGN channels so that $\mathbf B = \mathbf A$. It is observed form both Figs. \ref{fig_MSE_gamma_M} and \ref{fig_MSE_gamma} that when $\gamma$ is not very small, (i.e. when $\mathbf B$ is not very sparse), the impact of fading in recovering a given signal is not significant compared to that with AWGN channels. It is noted that the statistical properties of the measurement matrix changes from light-tailed to heavy-tailed when channels change from AWGN to Rayleigh fading. However, as $M$ and $\gamma$ increase, there is no significant difference in recovery performance with both types of channels.
Fig. \ref{fig_MSE_gamma} illustrates  the trade-off between $M$ and $\gamma$. It is seen that as $\gamma $  increases  beyond  $\approx 0.3$, the MSE performance decreases slowly with $\gamma$ for all values of $M$. This corroborates the theoretical results in (\ref{M_iid}) in which  the required number of MAC transmissions that enable recovery of $\mathbf x$ based on (\ref{l_1}) is proportional to $\frac{1}{\sqrt \gamma}$.

\begin{figure}
\centerline{\epsfig{figure=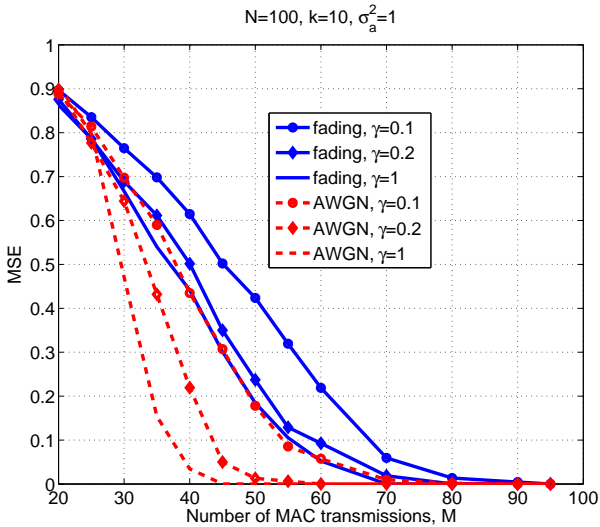,width=9.0cm}}
\caption{MSE vs number of MAC transmissions with nonidentical  fading channels; $\sigma_a^2=1$, $\nu_j\in [1,10] $ for all $j$, $N=100$, $k=10$ }
  \label{fig_MSE_M_noniidfading}
\end{figure}

\begin{figure*}[htb]
 \begin{minipage}[b]{.48\linewidth}
  \centering
  \centerline{\includegraphics[width=9.50cm]{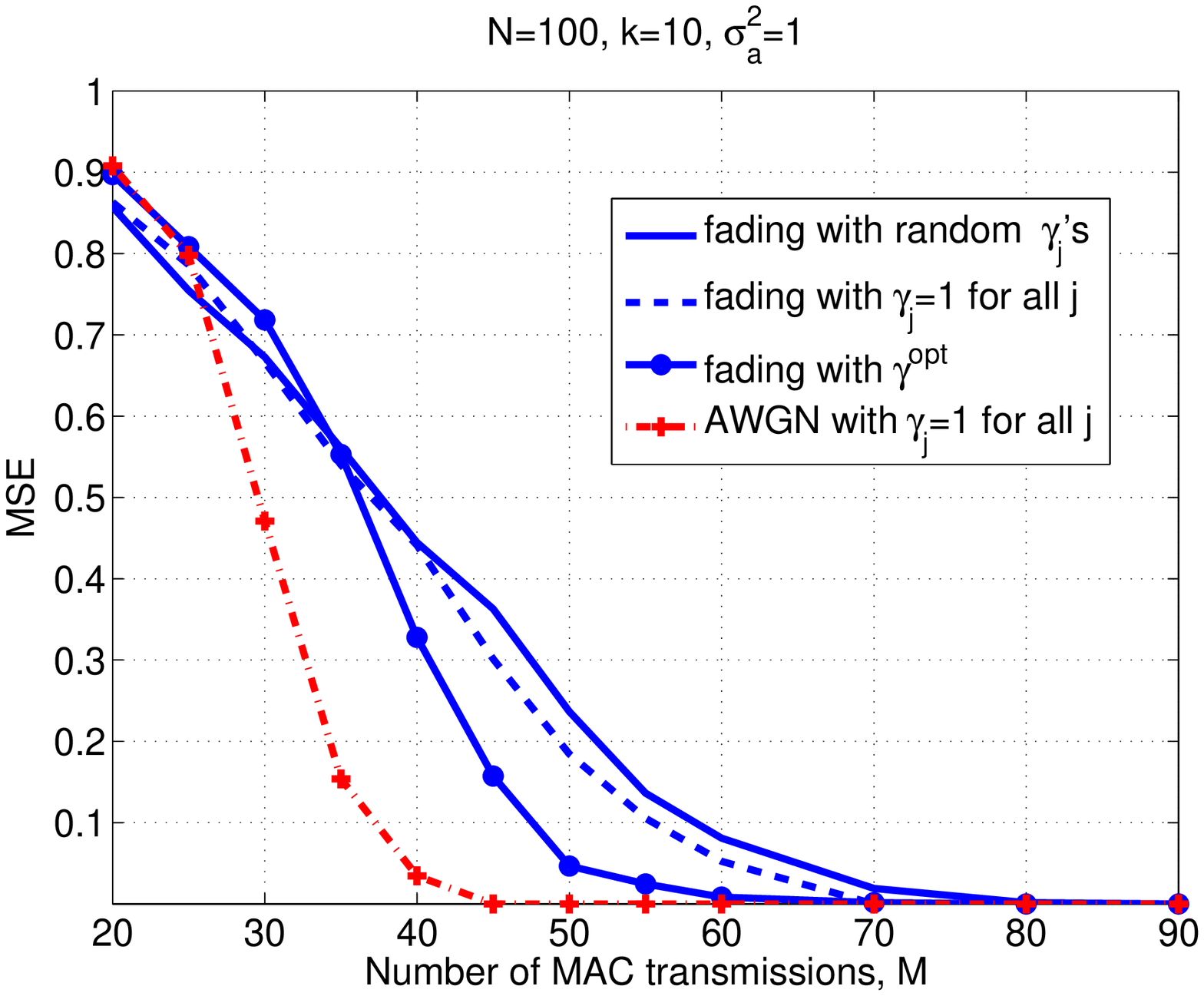}}
  \centerline{(a) $\nu_{\min} =1$, $\nu_{\max} = 10$}\medskip
\end{minipage}
\hfill
\begin{minipage}[b]{0.48\linewidth}
  \centering
  \centerline{\includegraphics[width=9.0cm]{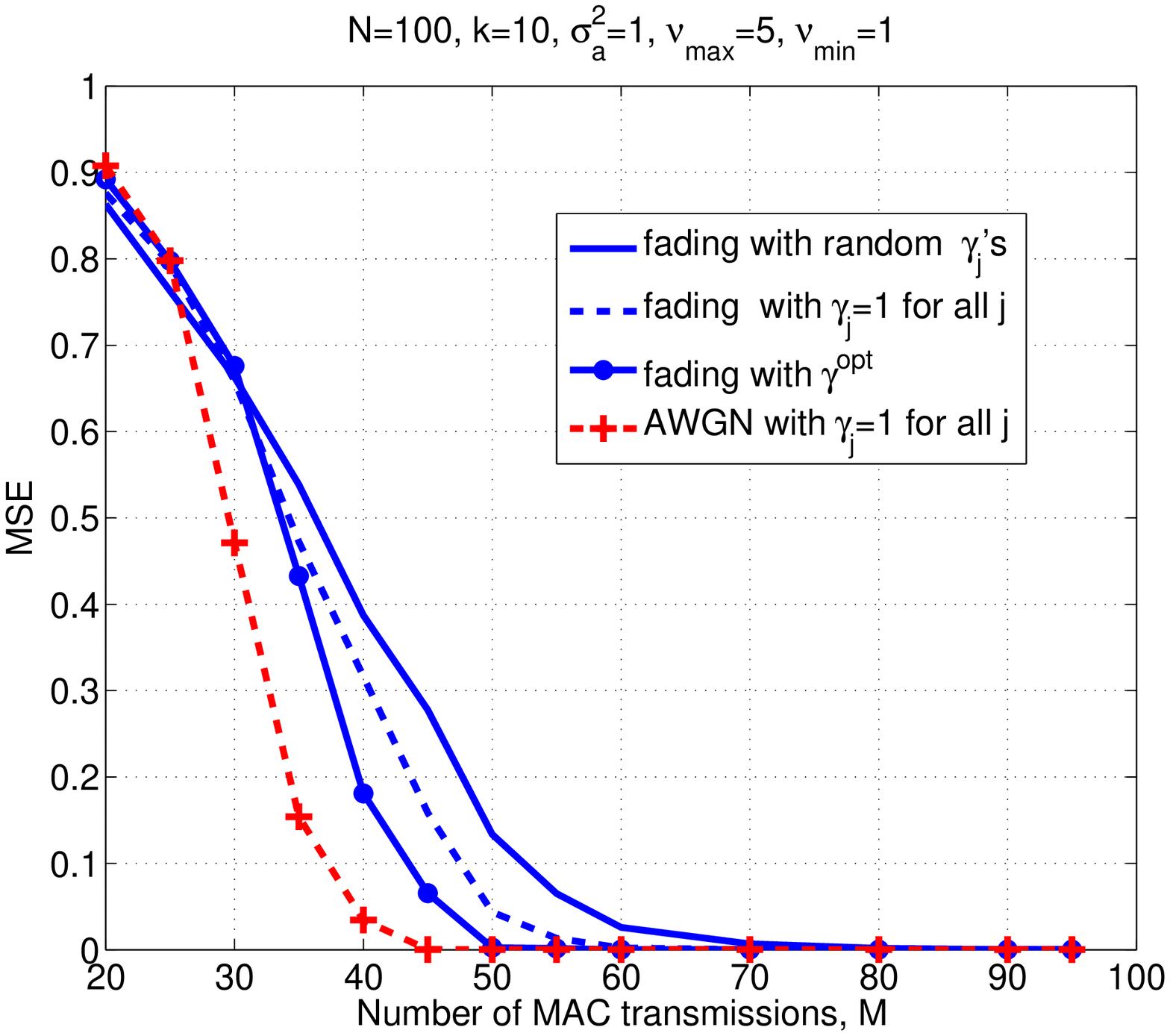}}
  \centerline{(b) $\nu_{\min} =1$, $\nu_{\max }= 5$}\medskip
\end{minipage}
\caption{MSE vs number of MAC transmissions with nonidentical  fading channels; $\sigma_a^2=1$, $N=100$, $k=10$ }
\label{fig_MSE_M_diffgamma_noniidfading}
\end{figure*}

\subsection{Nonidentical fading channels and identical  measurement sparsity parameters}
Next, we consider the case where fading channels are independent but nonidentical and the measurement sparsity parameter and the  power  at the each node are  the same over all the  nodes; i.e. $\gamma_j = \gamma$ and $\sigma_j^2 = \sigma_a^2$ for all $j=0,\cdots, N-1$. Each $\nu_j$ is selected uniformly from $[\nu_{\min},\nu_{\max}]$. It is noted that under this case, the elements of the matrix $\mathbf A$ are iid but those  in $\mathbf B$ are non iid. In Fig. \ref{fig_MSE_M_noniidfading}, we plot the MSE vs $M$.  We let $\nu_{\min}=1$, $\nu_{\max}=10$, $k=10$ and $N=100$. In contrast to Fig. \ref{fig_MSE_gamma_M} with iid fading channels, it can be seen from Fig. \ref{fig_MSE_M_noniidfading} that the presence of nonidentical fading channels reduces the capability of sparse recovery quite significantly compared to AWGN channels even when  the matrices are dense (i.e. $\gamma=1$). The reason is that,  under this case all the nodes transmit with equal probability irrespective of the quality of the fading channels leading to an inhomogeneous measurement matrix. On the other hand, with AWGN, the quality of all the channels is identical and the matrix $\mathbf A$ is isotropic. This will be further discussed in the  next section.

\subsection{Nonidentical fading channels and different measurement  sparsity parameters}
In this section,  we consider the case where  $\sigma_j^2 = \sigma_a^2$, and $\nu_j$'s  are selected uniformly from  $[\nu_{\min},\nu_{\max}]$  and arranged  in ascending  order for   $j=0,\cdots, N-1$. The values for $\gamma_j$'s are selected as in (\ref{gamma_opt}) so that the number of MAC transmissions is minimum  with respect to $\boldsymbol\gamma$.   We further assume $\frac{E}{\sigma_a^2} =1$ so that $\bar\gamma =1$ in (\ref{gamma_opt}).  In Fig. \ref{fig_MSE_M_diffgamma_noniidfading}, we plot the MSE vs $M$ with $\boldsymbol\gamma^{\mathrm{opt}}$. We further plot the performance   with $\gamma_j =1$ for all $j$ considering both AWGN and fading channels. In Fig. \ref{fig_MSE_M_diffgamma_noniidfading}(a) we let $\nu_{\min} = 1$ and $\nu_{\max} =10$ while in Fig. \ref{fig_MSE_M_diffgamma_noniidfading}(b), we have $\nu_{\min} = 1$ and $\nu_{\max} =5$. We make several important observation here. When $\gamma_j=1$ and $\sigma_j^2 = \sigma_a^2$ for all $j$, $\mathbf A$ is a iid Gaussian matrix, and $\mathbf H$ is a nonidentical (but independent) dense matrix with Rayleigh random variables. In that case, as seen in Fig. \ref{fig_MSE_M_diffgamma_noniidfading} , the recovery performance with $\mathbf B=\mathbf A\odot \mathbf H$ (shown in blue dash line) is significantly degraded compared to that with only $\mathbf A$ (shown in red marked  dash-dot line) due to the inhomogeneity of the matrix $\mathbf B$. This corroborates the theoretical results  shown in Section \ref{Dense_B}.     When the matrix $\mathbf A$ is made sparse with sparsity parameters as in (\ref{gamma_opt}), it can be seen that,  recovery performance (blue marked solid line) comparable to AWGN can be achieved especially the ratio $\frac{\nu_{\max}}{\nu_{\min}}$ is small. Further, when  the transmission probabilities are selected independent of $\nu_j$'s (i.e. randomly)  in the presence of nonidentical channel fading, a larger number of MAC transmissions is  necessary to achieve negligible MSE compared  to having only AWGN channels. When $M$ is small, it is observed that MSE with random $\boldsymbol\gamma$ is slightly smaller than that with  optimal gamma as found in (\ref{gamma_opt}).  It is worth mentioning that $M$ is optimized over $\boldsymbol\gamma$ considering perfect  signal recovery  and this optimality may not hold  when $M$ is very small (i.e. in the region where reliable signal recovery is not  guaranteed irrespective of $\boldsymbol\gamma$).

{ In Figures \ref{fig_MSE_gamma_M}-\ref{fig_MSE_M_diffgamma_noniidfading}, we assumed that the noise power at the fusion center is zero and the recovery is performed based on (\ref{l_1}). In Fig. \ref{fig_MSE_M_noniidfading_noise}, we consider the problem given in (\ref{l_12}) where the observations at the fusion center are noisy. Different values for the noise variance $\sigma_v^2$ are considered and $\epsilon_v$ is selected such that $\epsilon_v = \sigma_v\sqrt{M}\sqrt{1 + 2\sqrt 2/\sqrt M}$ \cite{Candes_l1magic2005}. As expected, it is observed from Fig. \ref{fig_MSE_M_noniidfading_noise} that as $\sigma_v^2$ increases, the recovery performance is degraded irrespective of the value of $\boldsymbol\gamma$. However, the use of  optimal $\gamma_j$, which is obtained in (\ref{gamma_opt}) considering the noiseless case,  improves the recovery  performance even when there is noise compared to the case where each node transmits  with arbitrary $\gamma_j$ for all $j$}.

\begin{figure}
\centerline{\epsfig{figure=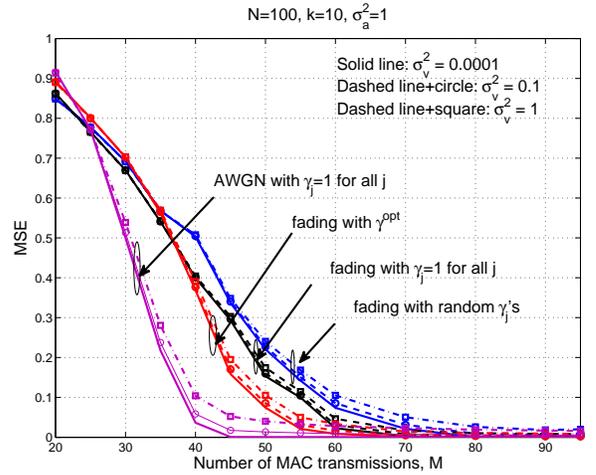,width=9.0cm}}
\caption{MSE vs number of MAC transmissions with nonidentical  fading channels and noise; $\sigma_a^2=1$, $N=100$, $k=10$, $\nu_{\min} = 1$, $\nu_{\max} = 10$}
  \label{fig_MSE_M_noniidfading_noise}
\end{figure}

\section{Conclusion}\label{conclusion}
In this paper, we considered the problem of sparse signal recovery in a distributed sensor network using sparse random  matrices. Assuming that  the channels between the sensors and the fusion center undergo fading, we derived  sufficient conditions that should be satisfied by the number of MAC transmissions to ensure recovery of a given $k$-sparse signal (i.e. for \emph{nonuniform} recovery). The impact of channel fading makes the corresponding random projection matrix heavy tailed with non identical elements compared to the widely used random matrices in the context of CS which are light tailed (sub-Gaussian). We have exploited the properties of subexponential random matrices with nonidentical elements in deriving nonuniform recovery guarantees under these conditions. We have shown that, when the channels undergo independent and  nonidentical fading, by properly  designing the probabilities of transmission at each node based on fading channel statistics, the number of measurements required for signal recovery can be reduced.  We further provided recovery guarantees of a given sparse signal when the projection matrix is  nonidentical sub-exponential in general. An interesting future work is to investigate the impact of channel interference on  sparse signal recovery in distributed networks.

\section*{Appendix A}
\subsection*{Proof of Proposition \ref{prop_1}}
Let $w=ha$ where we omit the subscripts of $h$ and $a$ for brevity. The pdf of $w$ is given by,
\begin{eqnarray}
f(w) = \int_{-\infty}^{\infty} \frac{h}{\sigma_h^2}e^{-\frac{h^2}{2\sigma_h^2}} \sqrt{\frac{1}{2\pi\sigma_a^2}} e^{-\frac{w^2}{h^2}\frac{1}{2 \sigma_a^2}} \frac{1}{|h|} dh \label{pdf_w_1}
\end{eqnarray}
Since $h\geq 0$, (\ref{pdf_w_1}) reduces to,
\begin{eqnarray*}
f(w) =\sqrt{\frac{1}{2\pi\sigma_a^2}}\frac{1}{\sigma_h^2} \int_{0}^{\infty} e^{-\frac{h^2}{2\sigma_h^2} -\frac{ w^2}{2 h^2\sigma_a^2}} dh
\end{eqnarray*}
Using the relation, $\int_0^{\infty} e^{-a x^2 - \frac{b}{x^2}} = \frac{1}{2} \sqrt{\frac{\pi}{a}} e^{-2 \sqrt{ab}}$ for $a,b > 0$, we have,
\begin{eqnarray*}
f(w) = \frac{1}{2 \bar{\sigma}}  e^{- {\frac{1}{\sigma_h\sigma_a}} |w|}
\end{eqnarray*}
where  $\bar{\sigma}=\sigma_a \sigma_h$,
which completes the proof.

\section*{Appendix B}
\subsection*{Proof of Proposition \ref{prop_mgf_sub_exp}}
When the pdf of $u$ is as given in (\ref{f_u}), we have,
\begin{eqnarray*}
\mathbb E \{e^{tu}\} &=& \int_{-\infty}^{\infty} e^{tu} \left(\gamma_j \frac{1}{2\bar\sigma_j} e^{-\frac{|u|}{\bar\sigma_j}} + (1-\gamma_j)\delta(u)\right) du\nonumber\\
&=& \frac{\gamma_j}{2\bar\sigma_j} \int_{-\infty}^0 e^{u\left(t+\frac{1}{\bar\sigma_j}\right)} du \nonumber\\
&+&  \frac{\gamma_j}{2\bar\sigma_j} \int_0^{\infty} e^{u\left(t-\frac{1}{\bar\sigma_j}\right)} du + (1-\gamma_j).
\end{eqnarray*}
When $|t| \leq  \frac{1}{\bar\sigma_j}$, it can be shown that
\begin{eqnarray}
\mathbb E \{e^{tu}\} = 1+ \gamma_j \frac{\sigma_j^2 t^2}{ 1- \bar\sigma_j^2 t ^2}.\label{E_etu}
\end{eqnarray}
This holds for any $\bar\sigma_j$ with $|t|  \leq  \frac{1}{\underset{j}{\max}\{\bar\sigma_j\}}$.
  Thus, when $  t ^2 < \frac{1}{\tilde\eta_{\max}^2}$, based on geometric series formula,  we have  $\frac{1}{ 1- \bar\sigma_j^2 t ^2} = \sum_{k=0}^{\infty}(\bar\sigma_j^2 t ^2)^k $. Thus, (\ref{E_etu})
can be approximated by,
\begin{eqnarray*}
\mathbb E \{e^{tu}\} &=& 1+ \gamma_j   \sum_{k=0}^{\infty}(\bar\sigma_j^2 t ^2)^{k+1} \nonumber\\
&\approx& 1+ \gamma_j {\sigma_j^2 t^2}\leq  e^{\gamma_j \bar\sigma_j^2 t^2} \leq  e^{\eta_{\max} t^2}
\end{eqnarray*}
where $\eta_{\max} =  \underset{j}{\max}\{\gamma_j \bar\sigma_j^2\}$,
completing the proof.

\section*{Appendix C}
\subsection*{Proof of Proposition \ref{prop_bern_inequ}}
Let $\Lambda = \sum_{i=0}^{N-1} \alpha_i u_i$. Then, using exponential Markov inequality, we have  \begin{eqnarray}
Pr(\Lambda \geq t)  &=& Pr(e^{\lambda \Lambda} \geq e^{\lambda t}) \leq e^{-\lambda t }\mathbb E\{e^{\lambda\Lambda}\}\nonumber\\
&=& e^{-\lambda t} \underset{i}{\prod} \mathbb E\{e^{\lambda \alpha_i u_i}\}. \label{Berntein_exp}
\end{eqnarray}
From Proposition \ref{prop_mgf_sub_exp}, we have,
\begin{eqnarray*}
Pr(\Lambda \geq t) \leq  e^{-\lambda t + \eta_{\max} \lambda^2 ||\boldsymbol \alpha||_2^2}
\end{eqnarray*}
when $|\lambda| \leq \frac{1}{\tilde\eta_{\max} ||\boldsymbol\alpha||_{\infty}}$.  Following similar steps as in the proof of Proposition $5.16$ in \cite{Vershynin_book1}, we get,
\begin{eqnarray*}
Pr(\Lambda \geq t) \leq  e^{-\min\left\{\frac{t^2}{4 \eta_{\max} ||\boldsymbol\alpha||_2^2}, \frac{t}{2\tilde\eta_{\max} ||\boldsymbol\alpha||_{\infty}}\right\}}.
\end{eqnarray*}
The same bound is obtained for $Pr(-\Lambda \geq t)$. Thus, we get (\ref{bern_noniid}).

\section*{Appendix D}
\subsection*{Proof of Theorem \ref{theorem_1}}

\begin{proof}(Theorem \ref{theorem_1})
 We follow similar proof techniques developed for nonuniform recovery with sub-Gaussian matrices in \cite{Ayaz_Tech1} with appropriate modifications to deal with sub-exponential random variables.
The failure probability of recovery of $\mathbf x$ based on (\ref{l_1}) is bounded by
\begin{eqnarray}
P_e &:=&  Pr(\exists l \notin \mathcal S |\langle (\mathbf B_S)^{\dagger} \mathbf b_l, \mathrm{sgn}(\mathbf x^S)\rangle| \geq 1)\nonumber\\
&\leq& (N-k) P_1^l < N P_1^l \label{P_e}
\end{eqnarray}
where   $P_1^l = Pr( |\langle (\mathbf B_S)^{\dagger} \mathbf b_l, \mathrm{sgn}(\mathbf x^S)\rangle| \geq 1 ) $.

To bound $P_1^l$, we use Proposition \ref{prop_bern_inequ}. Conditioned on $\mathbf B_{S}$, for given $l$ we have,
\begin{eqnarray*}
 P_1^l &=& Pr( |\langle (\mathbf B_S)^{\dagger} \mathbf b_l, \mathrm{sgn}(\mathbf x^S)\rangle| \geq 1 ) \nonumber\\
&=& Pr( |\langle  \mathbf b_l, (\mathbf B_S^{\dagger})^{\star} \mathrm{sgn}(\mathbf x^S)\rangle| \geq 1 )\nonumber\\
&=& Pr\left(|\sum_{j=0}^{M-1}  b_{l}(j) [(\mathbf B_{\mathcal S}^{\dagger})^* \mathrm{sgn}(\mathbf x^S) ](j)| \geq 1 \right)\nonumber\\
&\leq & 2 e^{- \min\left(\frac{1}{4 \eta_{\max} ||\mathbf b_S||_2^2}, \frac{1}{2 \tilde\eta_{\max} ||\mathbf b_S||_{\infty}}\right)}
\end{eqnarray*}
where we define $\mathbf b_S = (\mathbf B_{\mathcal S}^{\dagger})^* \mathrm{sgn}(\mathbf x^S)$ and $\eta_{\max}$ and $\tilde\eta_{\max}$ are as defined in Proposition \ref{prop_mgf_sub_exp}.  Thus,  we have
\begin{eqnarray}
P_e \leq N  2 e^{- \min\left(\frac{1}{4 \eta_{\max} ||\mathbf b_S||_2^2}, \frac{1}{2 \tilde\eta_{\max} ||\mathbf b_S||_{\infty}}\right)}.\label{Pe_bound2}
\end{eqnarray}
For $P_e$ in (\ref{Pe_bound2}) to be less than $\epsilon$,
we have to have,
\begin{eqnarray}
\min\left(\frac{1}{4 \eta_{\max} ||\mathbf b_S||_2^2}, \frac{1}{2 \tilde\eta_{\max} ||\mathbf b_S||_{\infty}}\right) \geq \log(2N/\epsilon).\label{minimum_exp}
\end{eqnarray}
We can see that (\ref{minimum_exp}) is satisfied when,
\begin{eqnarray*}
||\mathbf b_S||_2 \leq \frac{1}{2 \sqrt{\eta_{\max} \log(2N/\epsilon)}}
\end{eqnarray*}
and
\begin{eqnarray*}
||\mathbf b_S||_{\infty} \leq \frac{1}{2 \tilde\eta_{\max} \log(2N/\epsilon)}.
\end{eqnarray*}
It is noted that $\frac{||\mathbf b_S||_{\infty}}{||\mathbf b_S||_2}$ is the ratio between peak and total energy of $\mathbf b_S$ for given $\mathcal S$. Let $\frac{||\mathbf b_S||_{\infty}}{||\mathbf b_S||_2} \leq R$ where $0< R \leq 1$.

Thus,  (\ref{minimum_exp}) is satisfied when \begin{eqnarray}
||\mathbf b_S||_2  \leq \min\left(\frac{1}{2 \sqrt{\eta_{\max} \log(2N/\epsilon)}}, \frac{1}{2 \tilde\eta_{\max} R\log(2N/\epsilon)} \right). \label{condition_b_S}
\end{eqnarray}

Let  $\beta = \min\left(\frac{1}{2 \sqrt{\eta_{\max} \log(2N/\epsilon)}}, \frac{1}{2 \tilde\eta_{\max} R \log(2N/\epsilon)} \right) $.  To have $Pr(||\mathbf b_S||_2 \leq \beta) \geq 1- \epsilon'$ for $0< \epsilon'<1$, we have to have, $Pr(||\mathbf b_S||_2 \geq \beta) \leq \epsilon'$. To compute $P_2= Pr( ||\mathbf b_S||_2   \geq \beta)$, we use the following theorem.
\begin{theorem}[\cite{Vershynin_book1}] \label{theorem_2}Let $\mathbf A $ be a $M\times k$ matrix whose rows $\mathbf A_i$'s  are independent random vectors in $\mathbb R^k$ with the common second moment matrix $\Sigma = \mathbb E \{\mathbf A_i \otimes \mathbf A_i\}$. Let $T_0$ be a number such that $||\mathbf A_i||_2 \leq \sqrt{T_0}$ almost surely for all $i$. Then for every $t\geq 0$, the following inequality holds with probability at least $1-k e^{-c''t^2}$
\begin{eqnarray*}
||\frac{1}{M} \mathbf A^T \mathbf A - \Sigma|| \leq \max\{||\Sigma||^{1/2}\delta, \delta^2\}
\end{eqnarray*}
where $\delta = t \sqrt{\frac{T_0}{M}}$ and $c'' > 0$ is an absolute constant. Equivalently, we have,
\begin{eqnarray}
||\Sigma||^{1/2} \sqrt{M} - t \sqrt{T_0} \leq s_{\min}(\mathbf A) \leq s_{\max}(\mathbf A)\nonumber \\
\leq ||\Sigma||^{1/2} \sqrt{M} + t \sqrt{T_0}\label{eq_theorem2}
\end{eqnarray}
with probability at least $1-k e^{-c''t^2}$. Further, when $T_0 = \mathcal O(k)$, (\ref{eq_theorem2}) reduces to
\begin{eqnarray}
||\Sigma||^{1/2} \sqrt{M} - t \sqrt{k} \leq s_{\min}(\mathbf A) \leq s_{\max}(\mathbf A)\nonumber \\
\leq ||\Sigma||^{1/2} \sqrt{M} + t \sqrt{k}\label{eq2_theorem2}
\end{eqnarray}
with probability at least $1-k e^{-c't^2}$ where $c'$ is a constant.
\end{theorem}

Let $\bar{\mathbf B}_S = \Gamma_S \mathbf B_S$ where $\Gamma_S = \sqrt{\frac{k}{\sum_{j=0}^{k-1} 2 \gamma_{S_j}\bar{\sigma}_{S_j}^2}}$ and  $\gamma_{S_j}$ (similarly $\bar{\sigma}_{S_j}$) corresponds to $\gamma_i$ where $i=S_j$ is the $j$-th element of $\mathcal S$ for $j=0,\cdots, k-1$ and $i$ can take any value from $0, \cdots, N-1$.
It is noted that $P_2$ can be bounded by
\begin{eqnarray*}
P_2 &\leq& Pr\left(s_{\min}(\mathbf B_S) \leq \frac{\sqrt {k}}{\beta}\right)\nonumber\\
& =&
Pr\left(s_{\min}(\bar{\mathbf B}_S) \leq \Gamma_S \frac{\sqrt {k}}{\beta}\right).
\end{eqnarray*}
Let $\Sigma_{B_S} = \mathbb E\{(\bar{\mathbf B}_S)_i \otimes (\bar{\mathbf B}_S)_i \}$ be the second moment matrix of $(\bar{\mathbf B}_S)_i $ where $(\bar{\mathbf B}_S)_i  $ is the $i$-th row of the matrix $\bar{\mathbf B}_S$. Then, we have $\Sigma_{B_S} = \Gamma_S^2 \mathrm{diag}([2\gamma_{S_0}\bar{\sigma}^2_{S_0}, \cdots, 2\gamma_{S_{k-1}}\bar{\sigma}^2_{S_{k-1}}]^T)$. Since $\mathbf E\{||(\bar{\mathbf B}_S)_i||_2^2 \}^{1/2} = \sqrt{k}$, we can take $T_0 = \mathcal O(k)$ where $T_0$ is a number such that $||(\bar{\mathbf B}_S)_i||_2 \leq \sqrt{T_0}$ almost surely for all $i$. Thus, from Theorem \ref{theorem_2}, we have,
\begin{eqnarray*}
Pr\left(s_{\min}(\bar{\mathbf B}_S) \leq ||\Sigma_{B_S}||^{1/2} \sqrt{M} - t \sqrt{k}\right) \leq k e^{-c't^2}
\end{eqnarray*}
where $c'$ is a constant. Letting $t = \sqrt{\frac{M}{k}} \left(||\Sigma_{B_S}||^{1/2} - \frac{\Gamma_S}{\beta}\sqrt{\frac{k}{M}}\right)$, for $P_2\leq \epsilon'$, it is required that,
\begin{eqnarray}
\frac{M}{k}\left(||\Sigma_{B_S}||^{1/2} - \frac{\Gamma_S}{\beta}\sqrt{\frac{k}{M}}\right)^2 \geq \frac{1}{c'}\log\left(\frac{k}{\epsilon'}\right).\label{inequality_1}
\end{eqnarray}
After a simple manipulation, it can be shown that (\ref{inequality_1}) reduces to
\begin{eqnarray}
 \sqrt{\frac{M}{k}}\geq \frac{1}{||\Sigma_B||^{1/2}}\left(\frac{\Gamma_S}{ \beta} + \sqrt{\frac{\log(k/\epsilon')}{c'}}\right).\label{bound_1}
\end{eqnarray}
Using the relations, $\frac{1}{\sqrt{\underset{j}{\max} (2\gamma_j \bar\sigma_j^2)}} \leq \Gamma_S \leq \frac{1}{\sqrt{\underset{j}{\min} (2\gamma_j \bar\sigma_j^2)}}$ and $||\Sigma_B||^{1/2} \geq \sqrt{\frac{\underset{j}{\min} (2\gamma_j \bar\sigma_j^2)}{\underset{j}{\max} (2\gamma_j \bar\sigma_j^2)}}$, (\ref{bound_1}) is  satisfied when
\begin{eqnarray}
\sqrt{\frac{M}{k}}\geq \sqrt{\frac{\eta_{\max}}{\eta_{\min}}}\left(\frac{1}{\beta\sqrt{2\eta_{\min}} } + \sqrt{\frac{\log(k/\epsilon')}{c'}}\right)\label{M_lower_bound}
\end{eqnarray}
where we define $\eta_{\max} = \underset{j}{\max} (\gamma_j \bar\sigma_j^2)$ and $\eta_{\min} = \underset{j}{\min} (\gamma_j \bar\sigma_j^2)$.
When $\beta = \frac{1}{2 \sqrt{\eta_{\max} \log(2N/\epsilon)}}$, (\ref{M_lower_bound}) reduces to,
\begin{eqnarray*}
M \geq \frac{\eta_{\max}}{\eta_{\min}} k \left(\sqrt{\frac{2 \eta_{\max}}{\eta_{\min}}}\sqrt{\log(2N/\epsilon)} + \sqrt{\frac{\log(k/\epsilon')}{c'}}\right)^2
\end{eqnarray*}
On the other hand, when $\beta = \frac{1}{2 \tilde\eta_{\max}  R \log(2N/\epsilon)}$, we have, \begin{eqnarray*}
M \geq {\frac{\eta_{\max}}{\eta_{\min}}}  k \left(\frac{\sqrt 2 \tilde\eta_{\max}}{\sqrt{\eta_{\min}}} R \log(2N/\epsilon) + \sqrt M \sqrt{\frac{\log(k/\epsilon')}{c'}}\right)^2
\end{eqnarray*}

 completing the proof.
\end{proof}

\section*{Appendix E}
 We have
 \begin{eqnarray*}
 ||\mathbf b_S||_2^2& =&  (\mathrm{sgn}(\mathbf x^S))^T (\mathbf B_S^T \mathbf B_S)^{-1}  \mathrm{sgn}(\mathbf x^S)\nonumber\\
 &\approx& \frac{1}{M}\underset{j\in \mathcal S}{\sum} \frac{1}{2\gamma_j \bar\sigma_j^2}
 \end{eqnarray*}
 with sufficiently large $M$. Thus,
 \begin{eqnarray*}
 ||\mathbf b_S||_2 \geq \sqrt{\frac{k}{2M \eta_{\max}}}.
 \end{eqnarray*}
 With sufficiently large $M$, the $i$-th element of $\mathbf b_S$ can be approximated by,
 \begin{eqnarray*}
 \mathbf b_S(i) \approx \frac{1}{M} \sum_{j=0}^{k-1} \frac{(\mathbf  B_S)_{ij}}{2\gamma_{S_j}\bar\sigma_{S_j}^2} \mathrm{sgn}(\mathbf x^S)(j).
 \end{eqnarray*}
 Thus, we have,
 \begin{eqnarray*}
 | \mathbf b_S(i)| &\leq&  \frac{1}{M} \sum_{j=0}^{k-1} \frac{|(\mathbf  B_S)_{ij}|}{2\gamma_{S_j}\bar\sigma_{S_j}^2}.
 \end{eqnarray*}
 It is noted that $\mathbb E\{|(\mathbf  B_S)_{ij}|\} = \gamma_{S_j} \bar\sigma_{S_j}$. Thus,
  \begin{eqnarray*}
 \mathbb E\{| \mathbf b_S(i)|\} &\leq&  \frac{1}{M} \sum_{j=0}^{k-1} \frac{1}{2\bar\sigma_{S_j}}\nonumber\\
 &\leq&  \frac{k}{2M\tilde\eta_{\min}}
 \end{eqnarray*}
 where $\tilde\eta_{\min} = \underset{j}{\min}\{\bar\sigma_j\}$. Then we have,
 \begin{eqnarray*}
 \mathbb E\left\{\frac{||\mathbf b_S||_{\infty}}{||\mathbf b_S||_2} \right\} \leq \frac{\sqrt{k\eta_{\max}}}{\sqrt{2M\tilde\eta_{\min}^2}}\leq \frac{\sqrt{k} }{\sqrt{2M}}.
 \end{eqnarray*}
Thus, $R$ can be considered to be
 \begin{eqnarray*}
 R = \mathcal O\left(\sqrt{ \frac{k }{2M}}\right).
 \end{eqnarray*}

\section*{Appendix F}
\subsection*{Proof of Theorem \ref{theorem_general_sub}}
We follow a similar approach as in the proof of Theorem \ref{theorem_1} in Appendix D. For a given support set $\mathcal S$, the failure probability in recovering $\mathbf x$ from (\ref{l_1}) is upper bounded by,
\begin{eqnarray}
P_e \leq NP_l \label{Pe_general}
\end{eqnarray}
where $P_l = Pr(|\langle (\mathbf B_S)^{\dagger}\mathbf b_l, \mathrm{sgn}(\mathbf x^S)\rangle| \geq 1)$.
Using Proposition 5.16 in \cite{Vershynin_book1},  $P_l$ can be upper bounded by,
\begin{eqnarray*}
P_l \leq 2 e^{-c_1 \min\left( \frac{1}{\rho_{\max}^2 ||\mathbf b_S||_2^2}, \frac{1}{\rho_{\max} ||\mathbf b_S||_{\infty}}\right)}
\end{eqnarray*}
where $c_1$ is a constant and $\mathbf b_S = (\mathbf B_S^{\dagger})^* \mathrm{sgn}(\mathbf x^S)$ as defined in Appendix D. Then, $P_e$ in (\ref{Pe_general}) can be bounded above  by $\epsilon_1$ if
\begin{eqnarray}
 \min\left( \frac{1}{\rho_{\max}^2 ||\mathbf b_S||_2^2}, \frac{1}{\rho_{\max} ||\mathbf b_S||_{\infty}}\right) \geq \frac{1}{c_1}\log(2N/\epsilon_1). \label{condition_2}
\end{eqnarray}

Let the matrix $\mathbf B$ be such that for any given $\mathcal S$, $\frac{||\mathbf b_S||_{\infty}}{||\mathbf b_S||_2} \leq R_1$ almost surely where $0 < R_1 < 1$. Then, (\ref{condition_2}) is satisfied when,
\begin{eqnarray}
||\mathbf b_S||_2 \leq  \min\left( \frac{1}{\rho_{\max}}\sqrt{\frac{c_1}{\log(2N/\epsilon_1)}}, \frac{c_1}{\rho_{\max} R_1 \log(2N/\epsilon_1)}\right)\label{bs_2}
\end{eqnarray}
Let $\beta_1 =  \min\left( \frac{1}{\rho_{\max}}\sqrt{\frac{c_1}{\log(2N/\epsilon_1)}}, \frac{c_1}{\rho_{\max} R_1 \log(2N/\epsilon_1)}\right)$.
We have,
\begin{eqnarray*}
Pr (||\mathbf b_S||_2 \geq \beta_1) \leq Pr\left(s_{\min}(\mathbf B_S) \leq \frac{\sqrt k}{\beta_1}\right).
\end{eqnarray*}
Using the Theorem \ref{theorem_2}, we get,
\begin{eqnarray*}
Pr(s_{\min}(\mathbf B_S) \leq ||\Sigma_{\mathbf B_S}||^{1/2} \sqrt M - t \sqrt{T_0}) \leq k e^{-c'_1 t^2}
\end{eqnarray*}
where $T_0$ is a number such that $||(\mathbf B_S)_i||_2 \leq \sqrt{T_0}$ for all $i$ and $c'_1$ is a constant.
Letting $t= \sqrt{\frac{M}{T_0}} \left(||\Sigma_{\mathbf B_S}||^{1/2} - \frac{\sqrt k}{\beta_1\sqrt M}\right)$, it can be shown that $Pr (||\mathbf b_S||_2 \geq \beta_1) \leq \epsilon_1'$ when
\begin{eqnarray}
M  \geq \frac{1}{||\Sigma_{\mathbf B_S}||} \left( \frac{\sqrt k}{\beta_1} + \sqrt{\frac{T_0 \log(k/\epsilon'_1)}{c_1}}\right)^2. \label{M_general_sub_1}
\end{eqnarray}
Thus, (\ref{M_general_sub_1}) is satisfied when,
\begin{eqnarray*}
M  \geq \frac{1}{\lambda_{\min}} \left( \frac{\sqrt k}{\beta_1} + \sqrt{\frac{T_0 \log(k/\epsilon'_1)}{c'_1}}\right)^2
\end{eqnarray*}
where $\lambda_{\min}$ is the minimum eigenvalue of $\Sigma_{B}^T \Sigma_B$ completing the proof.

\bibliographystyle{IEEEtran}
\bibliography{IEEEabrv,bib1}

\end{document}